\documentclass[aip,cha,numerical,reprint,eqsecnum,floats]{revtex4-1} 

\usepackage[english]{babel}
\usepackage{amsmath}
\usepackage{amssymb}
\usepackage{mathrsfs}
\usepackage{amsthm}
\usepackage{textcomp}
\usepackage{xcolor}
\usepackage{graphicx}
\usepackage{enumerate}
\usepackage{mathtools}
\usepackage{hyperref}
\usepackage[title]{appendix}
\usepackage{bm}

\usepackage{mathrsfs}
\usepackage{amssymb}
\usepackage{amsmath}
\usepackage{amsfonts}
\usepackage{bm}
\usepackage{subfigure}
\usepackage{lipsum}
\usepackage{setspace}
\usepackage{amsthm}
\usepackage{graphicx}
\usepackage{mathtools}
\usepackage{xcolor}

\newtheorem{theorem}{Theorem}[section]
\newtheorem{lemma}[theorem]{Lemma}
\newtheorem{corollary}[theorem]{Corollary}
\newtheorem{prop}[theorem]{Proposition}
\theoremstyle{definition}

\theoremstyle{remark}
\newtheorem*{remark}{Remark}

\newcommand{\J}{{\cal J}}

\begin{document}

\title{Clusterization and phase diagram of the bimodal Kuramoto model with bounded confidence}

\author{Andr\'e Reggio}
\affiliation{School of Engineering, University of Applied Sciences of Western Switzerland, CH-1950 Sion, Switzerland}

\author{Robin Delabays}
\email{robin.delabays@hevs.ch}
\affiliation{School of Engineering, University of Applied Sciences of Western Switzerland, CH-1950 Sion, Switzerland}
\affiliation{Automatic Control Laboratory, Swiss Federal Institute of Technology (ETH), CH-8092 Z\"urich,  Switzerland}

\author{Philippe Jacquod}
\affiliation{School of Engineering, University of Applied Sciences of Western Switzerland, CH-1950 Sion, Switzerland}
\affiliation{Department of Quantum Matter Physics, University of Geneva, CH-1211 Geneva, Switzerland}
\date{\today}

\begin{abstract}
 Inspired by the Deffuant and Hegselmann-Krause models of opinion dynamics, we extend the Kuramoto model to account for confidence bounds, i.e., vanishing interactions between pairs of oscillators when their phases differ by more than a certain value. 
 We focus on Kuramoto oscillators with peaked, bimodal distribution of natural frequencies.
 We show that, in this case, the fixed-points for the extended model are made of certain numbers of independent clusters of oscillators, depending on the length of the confidence bound -- the interaction range --  and the distance between the two peaks of the bimodal distribution of natural frequencies. 
 This allows us to construct the phase diagram of attractive fixed-points for the bimodal Kuramoto model with  bounded confidence and to analytically explain clusterization in dynamical systems with bounded confidence. 
\end{abstract}

\maketitle

\begin{quotation}
There are many, very different physical systems that can be mathematically modelled by sets of coupled oscillators. 
When the coupling between individual oscillators is strong enough, the oscillators start to swing in unison, at the same frequency, even if their natural frequencies are quite different. 
The Kuramoto model is a popular model to describe this phenomenon of synchrony. 
Many physically important systems are characterized by couplings between pairs of oscillators that depend on the position of the latter. 
To investigate such systems, we construct an extension of the Kuramoto model, where the coupling vanishes for pairs of oscillators that lie too far away from one another. 
Focusing on a peaked bimodal distribution of natural frequencies, we show that the synchronous states are characterized by a clustering of the oscillators, where the system as a whole splits into disconnected subsystems that synchronize independently. The resulting partitioning allows us to calculate linear stability and domain of existence of these solutions, and to construct the phase diagram for the synchronous states vs. the system's parameters. 
\end{quotation}

\section{Introduction}

The scientific study of emergent phenomena is based on systems of interacting units - particles, agents, dynamical systems and so forth. 
Emergent phenomena are phenomena that cannot be comprehended from the individual dynamics of the interacting units. 
They would not occur without interaction and accordingly, they often cannot be anticipated by more microscopic approaches focusing on the system's components. 
One such emergent phenomenon is that of synchrony, where sufficiently strongly coupled oscillators start to swing in unison. Synchrony emerges out of the competition between the oscillator's natural tendency to swing at their own, natural frequency and the inter-oscillator coupling that tends to bind them together. 
To investigate this competition, Kuramoto introduced his celebrated model~\cite{Kur84,Ace05}
\begin{equation}\label{eq:kuramoto}
\dot\theta_i = \omega_i - \frac{K}{N} \sum_j \sin(\theta_i-\theta_j) \, .
\end{equation}
It describes $i=1, \ldots N$ oscillators of natural frequency  $\omega_i$ coupled to one another with a sinusoidal interaction of strength $K/N$.
When the latter exceeds a critical value the model synchronizes and the oscillators
rotate coherently at the same frequency. The critical interaction strength leading to synchrony depends on the distribution of natural frequencies.~\cite{Ace05}
Of particular interest to us in this manuscript are cases of bimodal distributions.~\cite{Mon04,DeSme08,Mar09}

Most systems of physical interest have however finite-range interactions. This motivated 
investigations of the Kuramoto model with couplings defined on incomplete graphs.~\cite{Ace05,Jad04}
A question of interest there is "how many phase-locked solutions to the Kuramoto model are there?", i.e., how many 
fixed-point solutions with $\dot\theta_i=0$ for all $i$ are there to Eq.~\eqref{eq:kuramoto}? 
It turns out that answering that question is hard, and depends on both the coupling network topology and the distribution of natural frequencies.~\cite{Kor72,Bai82,Jan03,Rog04,Och10,Ngu14,Meh15,Del16,Man17,Del17a,Del19} 
Another interesting aspect is that fixed-point solutions to Eq.~\eqref{eq:kuramoto} are characterized by their basin of attraction, and different solutions generically have basins with different volumes.~\cite{Wil06,Del17b,Del19}

Modelling a finite range of interaction can be achieved either by considering inter-oscillator couplings defined on incomplete graphs, or couplings depending on the phase of the oscillators. 
The Kuramoto model has been extended to account for coupling reduction when the phase difference between coupled oscillators increases,
\begin{align}\label{eq:kuramotoexp}
\dot\theta_i &= \omega_i - \frac{\alpha K}{N} \sum_j \sin(\theta_i-\theta_j) e^{-\alpha |\theta_i-\theta_j|} \, .
\end{align}
Because the model was applied to the problem of opinion synchronization with $\omega_i \ne 0$, it was named the {\it Opinion Changing Rate} model.~\cite{Plu06} 
With the sharp, exponential decrease of the interaction strength when oscillator phases differ, it is reminiscent of the Deffuant~\cite{Def00} and Hegselmann-Krause~\cite{Heg02} models, where agents interact only if their coordinate difference is below some threshold  $\Delta$, called the confidence bound.
The Hegselmann-Krause model~\cite{Heg02} is defined by
\begin{align}\label{eq:confidencebound}
\dot\theta_i &= - \frac{K}{N} \sum_j a_{ij} (\theta_i-\theta_j)  \, , 
\end{align}
with a coupling adjacency matrix $a_{ij}$ defined below in Eq.~\eqref{eq:kcb_bound}.~\cite{Lor07} 
Such models have been investigated mostly numerically and not many of their properties are known analytically. 
Their main feature is that their fixed-point solutions display clustering~\footnote{The literature sometimes refer to this phenomenon as \emph{fragmentation} or \emph{polarization}.} where agents coalesce on a finite number of different opinions. 
Analytically, convergence~\cite{Lor05} and convergence rates~\cite{Zhan15} has been established.
More recently, inhomogeneous models with agent-dependent confidence bounds have been investigated, with a focus on convergence properties.~\cite{Mir12,Cha17,Che20}

Motivated by these earlier works, we extend here the Kuramoto model to account for a finite, homogeneous confidence bound. 
In this initial work on this problem, we focus on a peaked bimodal distribution of natural frequencies. 
We determine and classify fixed-point synchronous solutions.
This is no trivial task: because of the confidence bound, solutions determine the coupling network, which in its turn determines the solutions. 
This self-consistency induces additional nonlinearities in the problem and in particular, it truncates basins of attraction in configuration space, because when phases change, so does the coupling network.
Our main results to be presented below are that (i) we identify all possible fixed-point solutions and determine their general structure, (ii) for most of these solutions, we identify their domain of existence in the model's parameter space, and (iii) we construct the phase diagram of the synchronous fixed points of the model.  
Along the way, (iv) we prove certain theorems on the linear stability of these fixed-point solutions, and (v) analytically derive clusterization in dynamical systems with bounded confidence, which so far had been established only numerically.~\cite{Lor07}

The paper is organized as follows. In Section~\ref{section1} we define the model. 
In Section~\ref{section2} we discuss general properties of its synchronous, fixed-point solutions. 
In particular we argue that they generally are made of several independent clusters.
In Section~\ref{sec:cluster} we characterize the structure of the clusters making up the synchronous states and identify two main types of clusters.
In Section~\ref{sec:clusterings} we come back to the configuration of fixed-points as collections of clusters. 
We identify which partitions of the total number of oscillators can give rise to a stable fixed point. 
Their domain of existence in parameter space is determined by the domain of existence of each of the clusters forming them. 
We discuss the linear stability and domain of existence of these fixed points. 
With these considerations we are able to construct the phase diagram of synchronous solutions to the Kuramoto model with bounded confidence. 
Conclusions are given in Section~\ref{conclusions}.

\section{The model}\label{section1}
We consider a model of $N$ Kuramoto oscillators with a dynamics given by~\cite{Kur84}
\begin{align}\label{eq:kcb_dyn}
 \dot{\theta_i} &= \omega_i - \frac{1}{N}\sum_{j=1}^Na_{ij}\sin(\theta_i-\theta_j) \, , 
\end{align}
where $\theta_i\in\mathbb{S}^1$ is the time-varying phase of the $i$th oscillator and $\omega_i\in\mathbb{R}$ is its natural frequency. 
The original Kuramoto model, Eq.~\eqref{eq:kuramoto}, considered all-to-all coupling with $a_{ij}=K$, $\forall \, i,j=1, \ldots N$. 
Here we extend it to account for a finite confidence bound $\Delta \ge 0$. 
Accordingly, the adjacency matrix elements of the coupling graph depend on phase differences and are defined as
\begin{align}\label{eq:kcb_bound} 
 a_{ij} &= \left\{
 \begin{array}{ll}
  1 \, ,& \text{if } \|\theta_j-\theta_i\| \leq \Delta\,, \\
  0 \, ,& \text{otherwise,}
 \end{array}
 \right. 
\end{align}
where $\|\theta_j-\theta_i\|$ gives the geodesic distance between $\theta_i$ and $\theta_j$ on $\mathbb{S}^1$. 
We depart 
somehow from the models of Eqs.\eqref{eq:kuramoto}--\eqref{eq:confidencebound} and absorb the parameter $K$ by a rescaling of time and of the natural frequencies.
The coupling $a_{ij}$ between oscillators $i$ and $j$ is symmetric and nonzero if and only if the geodesic distance between $\theta_i$ and $\theta_j$ is smaller than $\Delta$. 
We note that for $\Delta \geq \pi$, the original Kuramoto model with all-to-all coupling is recovered.~\cite{Kur84} 
The model defined by Eqs.~\eqref{eq:kcb_dyn} and \eqref{eq:kcb_bound} finds its inspiration in the Deffuant model of opinion dynamics,~\cite{Def00} but deviates from it in three important respects. 
First, the interaction between agents -- oscillators in our case --  is sinusoidal and not linear; second, the phase of all oscillators evolves continuously, as in Ref.~\onlinecite{Heg02}, whereas the Deffuant model updates only a single, randomly selected agent at each time step. 
In this sense, our model is closer to the Hegselmann-Krause model;~\cite{Heg02} third, the agents in our model have an intrinsic dynamics, which is absent in most models of opinion dynamics, with the exception of Ref.~\onlinecite{Plu06}.

In this first investigation of the Kuramoto model with confidence bound, we consider the simplest nontrivial distribution of 
natural frequencies, which is a symmetric peaked bimodal distribution
\begin{align}\label{eq:frq_propre_dist}
 \omega_i &= \left\{
 \begin{array}{ll}
  \omega_0\, , & \text{if } i \in {\bm I}_+\, , \\
  -\omega_0\, , & \text{if } i \in {\bm I}_-\, ,
 \end{array}
 \right.
\end{align}
with $\omega_0 \ge 0$, where $\bm{I}_+$ and $\bm{I}_-$ have the same cardinality and form a partition of the set of oscillators' indices. 
Our model is then characterized by two parameters, the confidence bound $\Delta\in[0,\pi]$ and the splitting $2 \omega_0$ between natural frequencies.

\section{General properties of synchronous states}\label{section2}
Depending on the strength of their diffusive coupling and their initial condition, Kuramoto models have a dynamic that brings them toward synchronous states where all oscillators have the same frequency, $\dot{\theta}_i=\dot{\theta}_j$ for all $i,j$. Our symmetric distribution of natural frequencies has $\sum_{i=1}^N\omega_i=0$, accordingly, all synchronous states are fixed-point solutions to Eqs.~\eqref{eq:kcb_dyn} and \eqref{eq:kcb_bound}, $\dot\theta_i=0$, for all $i$. 
They are therefore determined by their phases, modulo any homogeneous shift $\theta_i \rightarrow \theta_i+C$. 
We denote by $\bm{\theta}^*\in\mathbb{T}^n$ the equivalence class of such fixed-points and from now on we do not distinguish fixed points and their equivalence class. 
When $\|\theta_j^*-\theta_i^* \| \le \Delta $, $\forall~i,j$, the fixed-point synchronous state corresponds to a complete coupling graph, accordingly, it is the same as for the standard Kuramoto model.~\cite{Kur84} 
For that model, it is however well known that the distance between angles in the synchronous state depends on the distribution of natural frequencies, with the distance between angles increasing with the width of the distribution, see e.g., Ref.~\onlinecite{Str00}. 
Therefore, for fixed confidence bound $0<\Delta < \pi/2$, the coupling graph for such a state will no longer be complete, once $\omega_0$ exceeds some threshold value. 
Consequently, such a state may no longer be a fixed-point synchronous state for Eqs.~\eqref{eq:kcb_dyn} and \eqref{eq:kcb_bound}, while simultaneously new, different synchronous states may arise. 
In other words, for a given fixed-point $\bm{\theta}^*$, the coupling matrix $a_{ij}$ of Eq.~\eqref{eq:kcb_bound} depends on $\bm{\theta}^*$ which depends on $a_{ij}$ in its turn. 
This self-consistency is what makes this problem mathematically challenging and interesting. 

Generally speaking, the coupling matrix defined by a fixed-point  may correspond to a connected or to a nonconnected graph. 
In the former case, the problem is that of a Kuramoto model on a complete or non-complete graph, while the latter splits into subproblems, each of them defined by a connected component of the coupling graph. 
For fixed-point synchronous solutions to Eqs.~\eqref{eq:kcb_dyn} and \eqref{eq:kcb_bound}, each connected component corresponds to a set of oscillators that are synchronous and independent of oscillators in other connected components. 
We refer to these connected components as clusters. 
Consider then a fixed-point solution $\bm{\theta}^*$ corresponding to a disconnected coupling graph $a_{ij}$.
The oscillator phases and natural frequencies within each cluster $\Lambda$ of oscillators  must  satisfy
\begin{align}\label{eq:sumcluster}
 \sum_{i\in\Lambda}\omega_i - \frac{1}{N} \sum_{i,j\in\Lambda}a_{ij}\sin(\theta_i^*-\theta_j^*)=0\, ,
\end{align}
and because the coupling matrix is symmetric, the second term on the right-hand side must vanish. 
Hence $\sum_{i \in \Lambda} \omega_i=0$, accordingly, each cluster $\Lambda$ contains the same number of oscillators from $\bm{I}_+$, with natural frequency $\omega_0$ as from $\bm{I}_-$, with natural frequency $-\omega_0$. 
Fixed points therefore form sets of independent clusters, and accordingly correspond to certain partitions of the oscillators.
Because each acceptable partition consists of clusters with as many positive-frequency as negative-frequency oscillators, they are effectively partitions of $N$ oscillators in sets of even cardinality. 
We next discuss in detail what conditions these partitions need to satisfy. 

\begin{figure}
 \centering
 \includegraphics[width=\columnwidth]{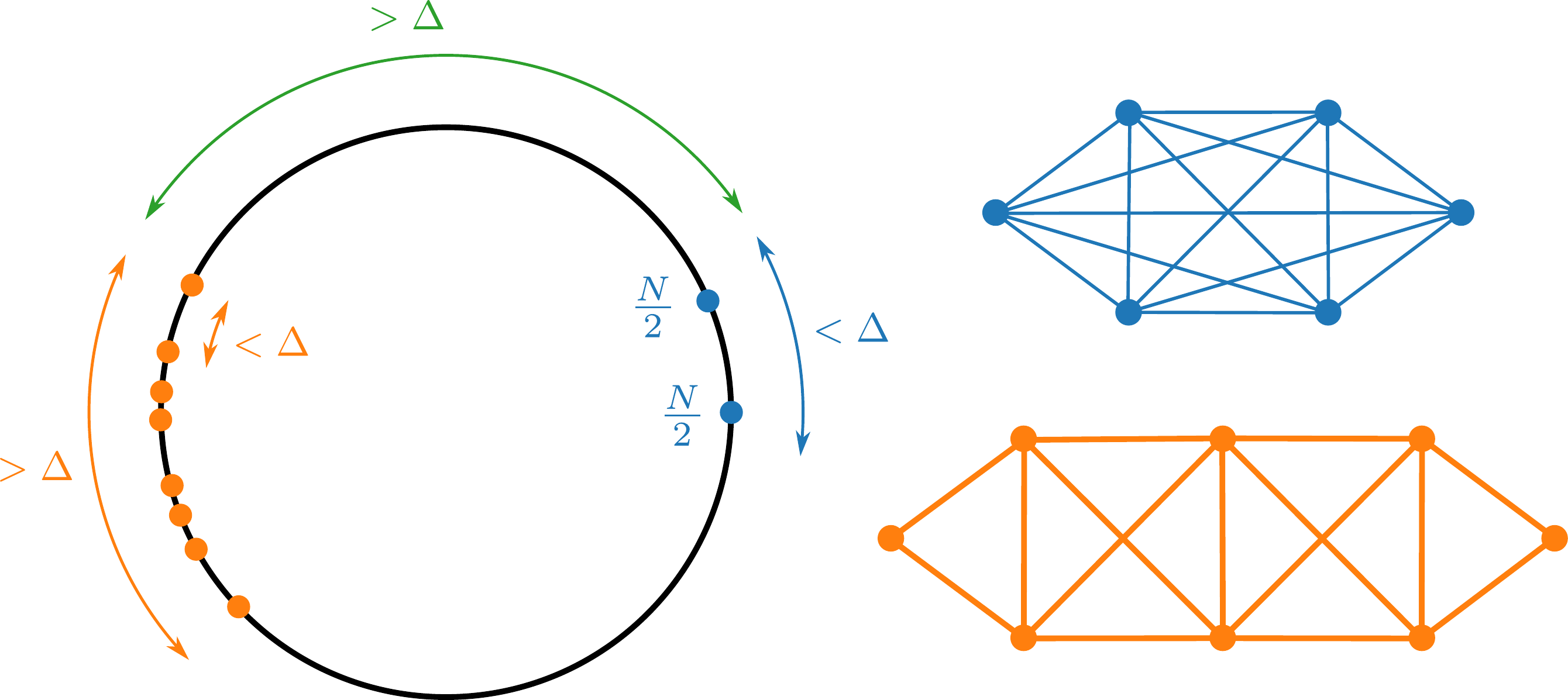}
 \caption{Illustration of the two types of clusters that we distinguish. 
 The blue oscillators all lie within an arc of length smaller than $\Delta$ and are aggregated in two groups of $N/2$ oscillators at the same phase. 
 They are then all connected with each other and form the all-to-all interaction graph illustrated in blue on the top right side. 
 If, within a group of oscillators, some of them are further away than $\Delta$ from each other, but they all belong to a connected interaction graph, then they form a sparse cluster, as illustrated in orange. 
 A fixed point can be formed of one or more clusters, each of them being either all-to-all or sparse. 
 If two clusters coexist, the arc length between their closest oscillators must be larger than $\Delta$ (green arrow). }
 \label{fig:illutration}
\end{figure}

\section{Synchronous states}\label{sec:cluster}
Synchronous states are made of clusters. 
Individual clusters satisfy Eq.~\eqref{eq:sumcluster} and therefore consist of two groups of equal numbers of oscillators, one with positive, the other one with negative natural frequencies. 
One may expect that all oscillators in each of these two groups have the same phases. 
We show that this is indeed the case for all-to-all clusters in stable fixed points in Section~\ref{section:ata}. 
Before that, we first consider such synchronized clusters and determine their domain of existence in the $(\Delta,\omega_0)$ parameter space. 
We next show that these are the only stable synchronous states for clusters with all-to-all coupling. 
Finally, while we cannot rule out the existence of different synchronous states for sparse coupling graphs, we show that the latter are systematically less stable and have a smaller domain of existence than synchronous clusters with all-to-all coupling. 

\subsection{2-group all-to-all clusters with peaked bimodal phases}\label{ssec:reg_ex}
We consider first a cluster with $2n$ oscillators on a complete coupling graph, in a configuration which we call \emph{2-group}, where the $n$ oscillators with natural frequency $\omega_0$ all have the same phase and the same holds for the $n$ oscillators with natural frequency $-\omega_0$. 
Such clusters are depicted in blue in Fig.~\ref{fig:illutration}.
We write
\begin{align}
\label{eq:2-groupconf}
 \theta_i &= \left\{
 \begin{array}{ll}
  \theta_+\, , & i\in \bm{I_+}\, , \\
  \theta_-\, , & i\in \bm{I_-}\, .
 \end{array}
 \right.
\end{align}
Because we assume that the coupling graph is complete, $||\theta_+-\theta_-|| \le \Delta$, and all other clusters have angles further away than a distance $\Delta$ from either $\theta_+$ or $\theta_-$ (green arrow in Fig.~\ref{fig:illutration}). 
From Eq.~\eqref{eq:kcb_dyn} one further has
\begin{align}\label{eq:state_ata}
 \omega_0 &= \frac{n}{N}\sin(||\theta_+-\theta_-||)\, .
\end{align}
Together, these two conditions demarcate the domain of existence of such a cluster in the $(\Delta,\omega_0)$ parameter space,
\begin{align}\label{eq:cc11}
 \omega_0 &\leq \left\{
 \begin{array}{ll}
  n \sin(\Delta)/N\, , & \text{if } \Delta\leq \pi/2\, , \\
  n/N\, , & \text{if } \Delta > \pi/2\, .
 \end{array}
 \right.
\end{align}

\subsection{Linear stability of fixed points made of 2-group clusters}\label{section:2group}
To be attractive, a fixed-point synchronous state needs to be linearly stable, that is, the Jacobian matrix of the coupling term of Eq.~\eqref{eq:kcb_dyn}, 
\begin{align}\label{eq.Jstability}
 \mathcal{J} _{ij} &= \left\{
 \begin{array}{ll}
  a_{ij}\cos(\theta_i^* - \theta_j^*) \, , & \text{if } i\neq j\, , \\
  -\sum_{k\neq i} a_{ik}\cos(\theta_i^* - \theta_k^*)\, , & \text{if } i= j  \, ,
 \end{array}
 \right. 
\end{align} 
is negative (semi-)definite in a neighborhood of the fixed point $\bm{\theta}^*$. 
For a fixed point with multiple clusters, $\cal{J}$ is a block-diagonal weighted Laplacian matrix, each block corresponding to one of its cluster. 
The spectrum of the Jacobian is the union of the spectra of its blocks, and the linear stability of each cluster, determined by the spectrum of each block, can be analyzed independently.

From Eq.~\eqref{eq.Jstability}, each block has the structure of a Laplacian matrix, with zero row and zero column sums of its elements. 
Then one of its eigenvalues vanishes. We order its eigenvalues as $0=\lambda_1 \ge \lambda_2 \geq...\geq\lambda_{2n}$, and the condition for stability is $\lambda_2 < 0$. 
We can write $\cal{J}$ as
\begin{subequations}\label{eq.Jacobian}
 \begin{align}
  \cal{J} &= 
  \begin{pmatrix}
   A & B \\
   B & A
  \end{pmatrix}
  \\
  A &= \bm{1}_{n \times n} - n \, (\cos(||\theta_+-\theta_-||)+1)\mathbb{I}_{n \times n} \\
  B &= \cos(||\theta_+-\theta_-||) \,  \bm{1}_{n \times n}\, ,
 \end{align}
\end{subequations} 
where we introduced the $n \times n$ matrix $(\bm{1}_{n \times n})_{ij}=1$ for all $i, j$ and the $n \times n$ identity matrix $\mathbb{I}_{n \times n}$. 
The following proposition shows that 2-group fixed-point clusters are linearly stable, from which the linear stability of synchronous states made of independent such clusters follows.

\begin{prop}\label{prop:eigen}
 The eigenvalues  of the stability matrix  
 associated with a cluster in the 2-group configuration are $\lambda_1=0$, $\lambda_2=-2n\cos(||\theta_+-\theta_-||)$, and $\lambda_{m\ge 3}=-n[\cos(||\theta_+-\theta_-||)+ 1]$. 
 The corresponding eigenvectors are  
 \begin{subequations}
 \begin{align}
  \bm{u}_1 &= (2n)^{-1/2}(1,...,1) \label{eq.lambda_1}\\
  \bm{u}_2 &= (2n)^{-1/2}(\underbrace{1,...,1}_{n},\underbrace{-1,...,-1}_{n}) \label{eq.lambda_2}\\
  \bm{u}_{m\ge 3} &= 2^{-1/2}(0,...,0,\underbrace{1,-1}_{{\rm at \,  position}\, (x,x+1)},0,...,0) \, , \label{eq.lambda_3} \\ & x \in\{1,...,n-1,n+1,...,2n\}\, \nonumber. 
 \end{align}
 \end{subequations} 

 A cluster in the 2-group configuration is linearly stable if and only if $||\theta_+-\theta_-|| < \pi/2$. 
\end{prop}

\begin{proof}
It is straightforward to see that the eigenmodes given above form a complete orthogonal basis. 
Using Eqs.~\eqref{eq.Jacobian}, a direct calculation shows that ${\cal J}   \bm{u}_m=\lambda_m \bm{u}_m$.

While $\lambda_{m\ge 3}$ is always negative, $\lambda_2$ is negative if and only if $||\theta_+-\theta_-|| <\pi/2$, which concludes the proof. 
\end{proof}

\begin{remark}
Because $\lambda_2$ is the only eigenvalue that can become positive, $\bm{u}_2$ gives the only possible direction of instability.
\end{remark}

We are interested in stable fixed points only. 
From now on, we will restrict to $||\theta_+-\theta_-|| \leq \pi/2$. 
This restriction together with Prop.~\ref{prop:eigen} implies that Eq.~\eqref{eq:state_ata} only has one solution 
\begin{align}\label{eq.phi_cluster}
||\theta_+-\theta_-||&= \arcsin\left(\omega_0\frac{N}{n}\right)\, ,
\end{align}
with cluster sizes $2 n\le N$ and accordingly
\begin{align}\label{eq:l2_ng}
    \lambda_2=-2n\sqrt{1-\big(\omega_0 \frac{N}{n}\big)^2}\, ,
\end{align}
which is valid as long as $\omega_0 N/n \le 1$.

\begin{corollary}
The eigenvalues and eigenvectors given in Prop.~\ref{prop:eigen} are also eigenvalues and -vectors for the stability matrix of an all-to-all Kuramoto 
model defined in Eq.~(\ref{eq:kcb_dyn}) with constant coupling $a_{ij}=1$, for all $i,j$ and the distribution (\ref{eq:frq_propre_dist}) of natural frequencies. 
For that model, Eqs.~(\ref{eq.phi_cluster}) and (\ref{eq:l2_ng}) mean that synchrony is lost when $2 \omega_0 > 1$.
\end{corollary}


\subsection{General all-to-all clusters}\label{section:ata}
We next consider all-to-all clusters without the 2-group condition. 
We show that synchronous fixed points solutions for such clusters are not linearly stable,
and are therefore not relevant for our goal of determining the $(\Delta,\omega_0)$ phase diagram of stable fixed points for the model defined in Eqs.~\eqref{eq:kcb_dyn}--\eqref{eq:frq_propre_dist}.

\begin{theorem}\label{thm:ata}
 Any linearly stable fixed-point solution for clusters of $2n$ oscillators with the dynamics determined by Eqs.~(\ref{eq:kcb_dyn})--(\ref{eq:frq_propre_dist}) and
 with all-to-all interactions is necessarily in the 2-group configuration defined in Eq.~(\ref{eq:2-groupconf}).
\end{theorem}

\begin{proof}
 Consider the standard Kuramoto order parameter~\cite{Kur84}
 \begin{align}\label{eq:oparam}
     re^{i\psi} &\coloneqq \frac{1}{2n}\sum_{j=1}^{2n}e^{i\theta_j}\, ,
 \end{align}
 restricted to the oscillators on the cluster. 
 As the interaction is all-to-all, Eq.~\eqref{eq:kcb_dyn} can be rewritten
 \begin{align}
     \dot{\theta}_i &= \omega_i - r\sin(\theta_i-\psi)\, ,
 \end{align}
 where $\psi$ can be set to zero by a redefinition of all angles, without loss of generality. 
 Once this is done, all phases $\theta_i^*$ of any fixed-point solution must satisfy
 \begin{align}
     \sin(\theta_i^*) &= \frac{\omega_i}{r}\, ,
     \label{eq:proof2groups}
 \end{align}
self-consistently with Eq.~\eqref{eq:oparam}.
 With our choice of natural frequencies $\omega_i=\pm \omega_0$, there are four angles $\theta^*\in[0,\pi/2]$, $\pi-\theta^*$, $-\theta^*$, and $-\pi+\theta^*$ that solve Eq.~\eqref{eq:proof2groups}. 
 Generic fixed-point solutions for clusters of $2n$ oscillators are therefore made of four groups of oscillators of sizes $n_1$, $n_2$, $n_3$, and $n_4$ at those angles. 
 This is illustrated in Fig.~\ref{fig:n1n2n3n4}. 
 
 \begin{figure}
     \centering
     \includegraphics[width=.8\columnwidth]{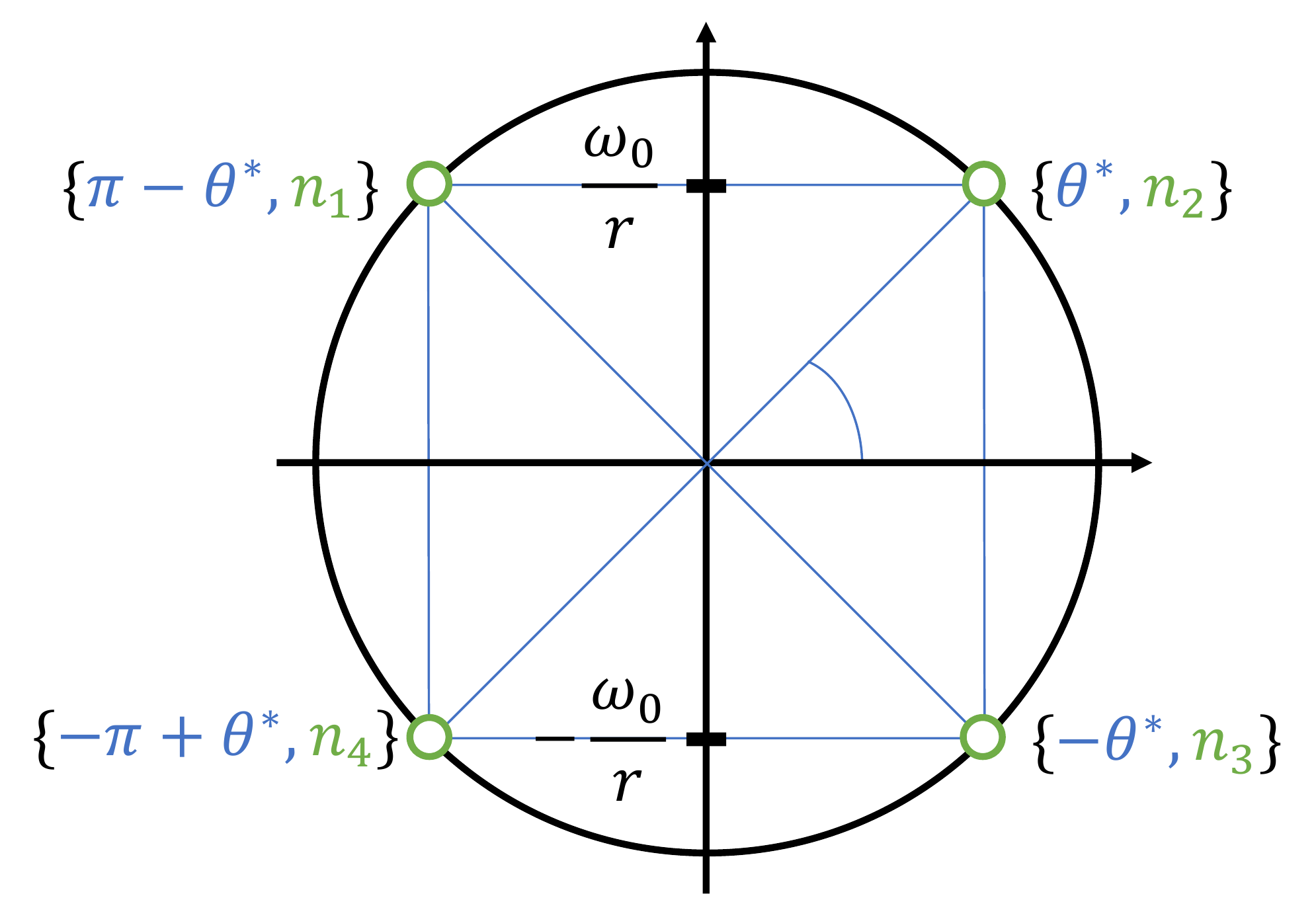}
     \caption{Location of the four solutions of Eq.~\eqref{eq:proof2groups} on the unit circle.}
     \label{fig:n1n2n3n4}
 \end{figure}
 
On can then write the Jacobian matrix in block form as
 \begin{align}\label{eq:jacobian}
     {\cal J} &= 
     \begin{pmatrix}
      D_1 & -c\bm{1}_{n_1\times n_2} & -\bm{1}_{n_1\times n_3} & c\bm{1}_{n_1\times n_4} \\
      -c\bm{1}_{n_2\times n_1} & D_2 & c\bm{1}_{n_2\times n_3} & -\bm{1}_{n_2\times n_4} \\
      -\bm{1}_{n_3\times n_1} & c\bm{1}_{n_3\times n_2} & D_3 & -c\bm{1}_{n_3\times n_4} \\
      c\bm{1}_{n_4\times n_1} & -\bm{1}_{n_4\times n_2} & -c\bm{1}_{n_4\times n_3} & D_4
     \end{pmatrix}\, ,
 \end{align}
 where $D_i=\bm{1}_{n_i\times n_i} + d_i \, \mathbb{I}_{n_i\times n_i}$, $c=\cos(2 \theta^*)$ and 
 \begin{align}
 \begin{split}
     d_1 &= (n_2-n_4)c + (n_3-n_1)\, ,  \\
     d_2 &= (n_1-n_3)c + (n_4-n_2) \, , \\
     d_3 &= (n_4-n_2)c + (n_1-n_3) \, , \\
     d_4 &= (n_3-n_1)c + (n_2-n_4)\, .
 \end{split}
 \end{align}
 
 Eq.~\eqref{eq:jacobian} shows that, unless $n_i n_j=0$ with $(i,j)=(1,3)$ and $(2,4)$, ${\cal J}$ has negative off-diagonal matrix elements, contributing a positive amount 
 to the diagonal matrix elements. 
 We are going to show that the presence of these positive contributions to the diagonal matrix elements results in at least one positive eigenvalue for ${\cal J}$. 
 Therefore the only stable fixed-points for clusters of $2n$ oscillators and all-to-all coupling are those with either $n_2=n_3=0$ or $n_1=n_4=0$, for $0 \le 2 \theta^* \le \pi/2$, or with either $n_1=n_2=0$ or $n_3=n_4=0$, for $\pi \ge 2 \theta^* \ge \pi/2$ since then all angle differences are smaller than $\pi/2$. 
 These are the 2-group cases treated above in Sec.~\ref{section:2group}.
 
 Assume that $n_1>0$ and $n_3>0$ and consider the matrix $-{\cal J}$. 
 Sylvester's criterion states that $-{\cal J}$ is positive semidefinite if and only if all its principal minors are nonnegative.~\cite{Hor86} 
 Take a principal minor as the determinant of the $2 \times 2$ matrix obtained by keeping the matrix elements of $-{\cal J}$ connecting one of the $n_1$ oscillators in the first set to one of the $n_3$ oscillators in the third set of oscillators and the corresponding diagonal elements. 
 One obtains
 \begin{align}
     \left[-{\cal J}\right]_{ij} 
     &= -\left[(n_4-n_2)c + (n_1-n_3)\right]^2 \leq 0\, .
 \end{align}
 This minor is negative, hence by Sylvester's criterion $-{\cal J}$ is not positive semidefinite and ${\cal J}$ is not negative semidefinite. 
 The corresponding fixed point is therefore unstable. 
 Similar arguments lead to the same conclusion, unless the 2-group conditions, $n_2=n_3=0$ or $n_1=n_4=0$, for $0 \le 2 \theta^* \le \pi/2$, or with either $n_1=n_2=0$ or $n_3=n_4=0$, for $\pi \ge 2 \theta^* \ge \pi/2$, are satisfied. 
 This concludes the proof. 
\end{proof}

By Theorem~\ref{thm:ata}, the only stable all-to-all clusters are in the 2-group configuration. 
The remaining case of clusters that interact on sparse, incomplete coupling graphs is discussed in the next paragraph.

\subsection{Sparse clusters}
Numerical results show that relatively rare initial conditions for the bimodal Kuramoto model with confidence bound defined in Eqs.~\eqref{eq:kcb_dyn}--\eqref{eq:frq_propre_dist} lead to fixed points where oscillators have clustered phases with a spread exceeding $\Delta$. 
Such clusters are illustrated in orange in Fig.~\ref{fig:illutration}. 
When this is the case, clusters correspond to connected but incomplete graphs. 
We call them \emph{sparse clusters}. 
These clusters are harder to analyze, in particular the analysis leading to Eq.~\eqref{eq:state_ata} and the identification of the domain of existence of all-to-all fixed points cannot be applied to sparse clusters. 
We have been unable to include sparse clusters into the phase diagram for stable fixed points, however we observed numerically that they are rare in the sense that they 
attract only few initial phase configurations. 
Here we shed light on this numerical observation by showing that sparse clusters have (i) smaller domain of existence and (ii) lower stability than all-to-all clusters. 

We first need the following lemma, that relates the angle difference $||\theta_+-\theta_-||$ in all-to-all clusters to the angle difference between oscillators corresponding to the natural frequency $\omega_0$ and those for $-\omega_0$ in sparse clusters. 

\begin{lemma}\label{lem.sinphi}
 Let $\bm{\theta}^*\in\mathbb{R}^{2n}$ be a fixed point for Eqs.~(\ref{eq:kcb_dyn})--(\ref{eq:frq_propre_dist}) in a sparse cluster configuration. 
 Assume that in a certain range of parameters $(\Delta,\omega_0)$ it coexists with a fixed point in the all-to-all configuration, and let $||\theta_+-\theta_-||$ be the angle difference for that fixed point. 
 Then
 \begin{align}
  \sin(||\theta_+-\theta_-||)=\frac{1}{n^2}\sum_{i\in \mathbf{I}_+}\sum_{j\in\mathbf{I}_-}a_{ij}\sin(\theta_i^*-\theta_j^*)\eqqcolon S\, ,
  \label{eq.sinphi}
 \end{align}
 with $a_{ij}$ defined in Eq.~(\ref{eq:kcb_bound}).
\end{lemma}

\begin{proof}
 Comparing Eq.~\eqref{eq:kcb_dyn} for both fixed points  gives 
 \begin{align}\label{eq:balance}
  n\sin(||\theta_+-\theta_-||) &= \sum_j a_{ij}\sin(\theta_i^*-\theta_j^*)\, ,
 \end{align}
 for all $i\in\bm{I}_+$. 
 Summing Eq.~\eqref{eq:balance} over $i\in \mathbf{I_+}$ and noting that
 $\sum_{i\in \mathbf{I}_+}\sum_{j\in\mathbf{I}_+}a_{ij}\sin(\theta_i^*-\theta_j^*)=0$ by symmetry gives Eq.~\eqref{eq.sinphi} and concludes the proof.
\end{proof}

With Lemma~\ref{lem.sinphi} at hand, we can show that fixed points corresponding to all-to-all clusters have the largest domain of existence in the $(\Delta,\omega_0)$ parameter space.

\begin{theorem}
 For the same number of oscillators $2n$, the domain of existence of stable fixed points in the all-to-all configuration contains the domain of existence of stable fixed points in a sparse configuration. 
\end{theorem}

\begin{proof}
Eq.~\eqref{eq.sinphi} states in particular that at least one angle difference between connected oscillators in the sparse fixed point is larger than the angle difference $||\theta_+-\theta_-||$ in the all-to-all fixed point. 
This is so because sparsity implies that at least one $a_{ij}=0$ on the right-hand side of Eq.~\eqref{eq.sinphi}, while for pairs of connected oscillators one has $a_{ij}=1$. 
For both fixed points, the confidence bound $\Delta$ must exceed the largest angle difference between pairs of connected oscillators. 
Therefore, everything else being fixed, the smallest admissible value of $\Delta$ for a sparse configuration is larger than the smallest admissible $\Delta$ for the all-to-all configuration, i.e., the domain of existence of the sparse fixed point is contained in the domain of existence of the all-to-all fixed point.  
\end{proof}

We next show that fixed points corresponding to all-to-all clusters are more stable than those corresponding to sparse clusters. 

\begin{theorem}
 For the same number of oscillators $2n$, the largest non-vanishing eigenvalue of the Jacobian [Eq.~(\ref{eq.Jstability})] is smaller for the stable fixed point in the all-to-all configuration than for a stable fixed point in a sparse configuration. 
\end{theorem}

\begin{proof}
Let ${\cal J}$ and ${\cal J}'$ be the Jacobian matrices of Eq.~\eqref{eq:kcb_dyn} corresponding to the fixed points respectively in all-to-all and sparse configurations. 
Let $0=\lambda_1 > \lambda_2 \geq ... \geq \lambda_{2n}$, $\bm{u}_1$,...,$\bm{u}_{2n}$ and $0=\lambda_1' > \lambda_2' \geq ... \geq \lambda_{2n}'$, $\bm{u}_1'$,...,$\bm{u}_{2n}'$ be their eigenvalues and eigenvectors. 
In particular, the eigenvalues of ${\cal J}$ are given in Prop.~\ref{prop:eigen} as a function of the phase difference $||\theta_+-\theta_-||$ [see Eq.~\eqref{eq.phi_cluster}]
between oscillators with positive and negative natural frequencies.
The phase coordinates of the sparse fixed point are $\bm{\theta}^*\in\mathbb{R}^{2n}$.

Because ${\cal J}'$ is real symmetric, each of the above two sets of eigenvectors form an orthonormal basis of $\mathbb{R}^{2n}$. 
This implies that for any vector $\bm{v}$ such that $\bm{v} \perp \bm{u}_1=\bm{u}_1'=(2n)^{-1/2}(1,...,1)$, 
\begin{align}
 \frac{\bm{v}^\top \J'\bm{v}}{||\bm{v}||^2}\leq \lambda_2'\,.
 \label{eq.prooflambda2}
\end{align}
In particular, this is true for $\bm{v}=\bm{u}_2$. 
Thus if we can prove that 
\begin{align}\label{eq.proofstability}
 \bm{u}_2^\top \J'\bm{u}_2 &> \lambda_2\, ,
\end{align}
then we simultaneously prove that $\lambda_2' \ge \lambda_2$.

With Eqs.~\eqref{eq.Jstability} and \eqref{eq.lambda_2},
a straightforward calculation shows that 
\begin{align}\label{eq:fin_sum}
 \bm{u}_2^\top\J'\bm{u}_2 &= -\frac{2}{n}\sum_{i\in\bm{I}_+}\sum_{j\in\bm{I}_-}a_{ij}\cos(\theta_i^*-\theta_j^*) =: -2 n C \, . \end{align}
Because we consider stable fixed point, ${\cal J}'$ is negative semidefinite, and therefore $C \ge 0$. 
Combining Eq.~\eqref{eq:fin_sum} and the expression of $\lambda_2$ in Prop.~\ref{prop:eigen}, proving Eq.~\eqref{eq.proofstability} is equivalent to proving
\begin{align} \label{eq.Cdef}
 \cos(||\theta_+-\theta_-||) &> \frac{1}{n^2}\sum_{i\in\mathbf{I_+}}\sum_{j\in\mathbf{I_-}}a_{ij}\cos(\theta_i^*-\theta_j^*) = C\, .
\end{align}
Because $C \ge 0$, this inequality can be rewritten as
\begin{align} \label{eq.proofstability2}
 C^2 + S^2 &< 1 \, ,
\end{align}
where $S$ is defined in Eq.~\eqref{eq.sinphi}. Eq.~\eqref{eq.proofstability2} is finally verified 
by direct computation,
\begin{align}\label{eq:c2s2}
C^2 + S^2
 &=  n^{-4} \sum_{\substack{i,k\in\bm{I}_+ \\ j,l\in\bm{I}_-}} a_{ij}a_{kl} 
 \cos\left[(\theta_{i}^* - \theta_{j}^*) - (\theta_{k}^* - \theta_{l}^*)\right]\, ,
\end{align}
which is necessarily smaller than $1$. 
By equivalence of Eqs.~\eqref{eq.proofstability2}, \eqref{eq.Cdef}, and \eqref{eq.proofstability}, we conclude that $\lambda_2' \ge \lambda_2$, which concludes the proof.
\end{proof}

\begin{remark}
Eq.~\eqref{eq:c2s2} takes value $1$ only for the all-to-all cluster, for which $a_{ij}=1$, and $(\theta_{i}^* - \theta_{j}^*) - (\theta_{k}^* - \theta_{l}^*)=0$
for $i,k \in \bm{I}_+$ and $j,l\in\bm{I}_-$. 
\end{remark}

This concludes our discussion of independent clusters of oscillators in stable fixed points. 
Fixed points of the Kuramoto model with confidence bound, defined in Eqs.~\eqref{eq:kcb_dyn}--\eqref{eq:frq_propre_dist}, are formed of independent clusters. 
We know from numerical data that some of these fixed points are made of sparse clusters, but most of them are made of all-to-all clusters within which all pairs of oscillators are coupled. 
For the latter, since oscillators come in pairs of as many oscillators with positive as with negative natural frequencies, stable fixed points are characterized by only two phase values, depending only on the two natural frequencies in our bimodal Kuramoto model. 
The difference between the two phase values is given in Eq.~\eqref{eq.phi_cluster}. 
This allows us to demarcate the domain of existence of these fixed points and to compute the value of the largest negative eigenvalue of the corresponding stability matrix, giving a measure of the fixed point linear stability. 

Unfortunately, we are not able at this stage to demarcate the domain of existence of sparse clusters, nor to evaluate the largest eigenvalue of their stability matrix. 
Still we are able to prove that their domain of existence in the $(\Delta,\omega_0)$ parameter space is smaller and that they are less stable, in the sense of the largest negative eigenvalue of the stability matrix, than all-to-all clusters with the same number of oscillators. 
This explains our numerical observation that fixed points with sparse clusters are rare. 
From here on we focus on all-to-all clusters.

\section{Partitions}\label{sec:clusterings}
Fixed points are in general made of several clusters and accordingly correspond to partitions of the total number $N$ of oscillators in the system. 
Because of our choice of a peaked bimodal distribution of natural frequencies, $\omega_i = \pm \omega_0$, we know that each cluster in this partition consists of an even number $2 n_\alpha$ of oscillators, with $\alpha$ indexing the clusters. 
A fixed point is therefore represented by a partition of its $N$ oscillators into subsets $\Lambda_\alpha$, with $\sum_\alpha 2 n_\alpha = N$. 
We now consider such partitions into all-to-all clusters and evaluate their domain of stability. 

\subsection{Linear stability}
The spectrum of the Jacobian / stability matrix for a fixed point made of $N_c$ independent clusters is given by the union of the spectra of the Jacobians of  each 
independent cluster. 
The linear stability of the fixed point can then be measured by the largest eigenvalue $\lambda_2^{(\alpha)}$ corresponding to its 
least stable cluster $\alpha$. Eq.~\eqref{eq:l2_ng} states that this is the cluster with smallest number $n_\alpha$ of oscillators. To compare the stability
of two coexisting but different fixed points, one therefore looks for the smallest clusters with different sizes in both partitions. 
A numerical example will be given for the case of $N=12$ oscillators, where the fixed point with a single cluster of $2n=N$ oscillators is the most stable, 
while the partition with six clusters of $2n=2$ oscillators is the least stable.

\subsection{Domain of existence}
Eq.~\eqref{eq:cc11} demarcates the domain of existence of a single cluster. 
A fixed point corresponds to a partition into $N_c$ clusters, each of them having its limited domain of existence -- it must be able to exist individually -- and each cluster being sufficiently far away from any other one so as to not interact with it. 

The condition that each cluster must be able to exist independently limits the value of $\omega_0$ with respect to the smallest cluster size as
\begin{align}\label{eq:CC1}
 \omega_0 &\leq \left\{
 \begin{array}{ll}
  \min_\alpha [f_{\alpha}] \sin(\Delta)\, , & \text{if } \Delta\leq \frac{\pi}{2}\, , \\
  \min_\alpha [f_{\alpha}] \, , & \text{if } \Delta > \frac{\pi}{2}\, ,
 \end{array}
 \right. 
\end{align}
where $2 f_\alpha = 2 n_\alpha/N \in[0,1]$ gives the fraction of oscillators in cluster $\alpha$. 
The condition that each cluster does not interact with any other one reads 
\begin{align}\label{eq:cc2}
 \min_{i\neq j}D(\Lambda_\alpha,\Lambda_\beta) &> \Delta\, , & \forall& \alpha, \beta \, , 
\end{align}
where  the distance between two clusters labeled $\Lambda_\alpha$ and $\Lambda_\beta$ is
\begin{align}\label{eq:cluster_dist}
 D(\Lambda_\alpha,\Lambda_\beta) &\coloneqq \min_{i\in \Lambda_\alpha, j \in \Lambda_\beta} ||\theta_i^*-\theta_j^*|| \, .
\end{align}
Eq.~\eqref{eq:l2_ng} shows that the phase difference $||\theta_{\alpha,+}^*-\theta_{\alpha,-}^*||$ giving the angular spread of each cluster increases with $\omega_0$ so that the maximal value for $\omega_0$ still accepting the coexistence of a partition is given by  
\begin{align}\label{eq:cc_homog}
 N_{c}^{-1}\left(2\pi-\sum_{\alpha=1}^{N_c} ||\theta_{\alpha,+}^*-\theta_{\alpha,-}^*|| \right) &> \Delta\, .
\end{align}
With Eq.~\eqref{eq.phi_cluster} we finally get a second condition for the domain of existence in the $(\Delta,\omega_0)$ parameter space 
of the fixed point with $N_c$ clusters,
\begin{align}\label{eq:CC2}
 N_{c}^{-1}\left(2\pi-\sum_{\alpha=1}^{N_c} \arcsin\left(\omega_0/f_\alpha\right)\right) &> \Delta\, .
\end{align}

Eqs.~\eqref{eq:CC1} and \eqref{eq:CC2}, together with the general stability condition $||\theta_{\alpha,+}^*-\theta_{\alpha,-}^*||\leq \pi/2$, $\alpha=1,...,N_c$
demarcate the domain of existence of a fixed point characterized by a given partition. 
Fig.~\ref{fig:EA} shows the phase diagram for the existence of different fixed points characterized by different partitions in the $(\Delta,\omega_0)$
parameter space for a system of $N=12$ oscillators. 

Our findings are valid in the limit $\omega_0\searrow 0$, where the two groups forming each cluster coalesce. 
The bottom line of Fig.~\ref{fig:EA} shows the multiplicity of possible clusterings with respect to $\Delta$ in the limit $\omega_0\searrow 0$. 
It gives an analytical description of the $\Delta$-dependent clustering that has been observed numerically in similar models.~\cite{Lor07} 

\begin{figure}
 \centering
 \includegraphics[width=8cm]{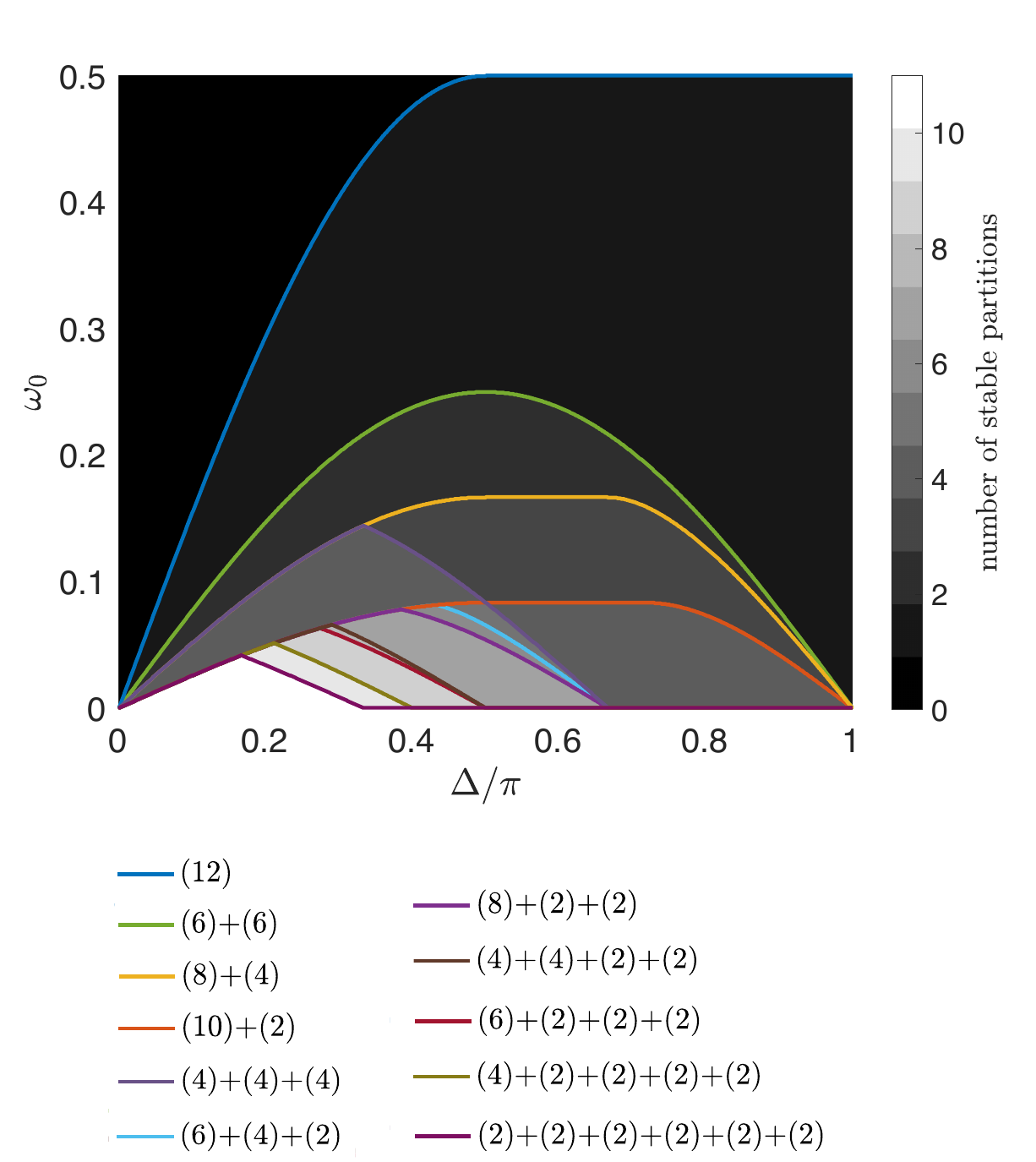}
 \caption{Domains of existence of the stable fixed points and the corresponding partitions of the bimodal Kuramoto model with confidence bound defined in Eqs.~\eqref{eq:kcb_dyn}--\eqref{eq:frq_propre_dist} with $N=12$ oscillators. 
 Each colored line gives the upper boundary of the existence domain of the corresponding partition. 
 The grey scale indicates the number of different possible fixed points in the 2-group configuration for each area. }
 \label{fig:EA}
\end{figure}

A final result allows one to anticipate which partition has the largest domain of existence, without resorting to a direct calculation of the boundaries of the domains. 

\begin{prop}\label{prop:largest domain of existence}
For a fixed number of clusters $N_c$, the fixed point of the bimodal Kuramoto model with confidence bound defined in Eqs.~(\ref{eq:kcb_dyn})--(\ref{eq:frq_propre_dist}),
corresponding to the partition where all clusters have the same number of oscillators, has the largest domain of existence. 
\end{prop}

\begin{proof}
First, Eq.~\eqref{eq:CC1} states that, to minimize the constraints on  $\omega_0$ vs. $\Delta$, the smallest cluster has to be as large as possible. 
This is obtained for an homogeneous size of the clusters with $f_\alpha=1/2N_c$ for all $\alpha$. 
The region delimited by Eq.~\eqref{eq:CC1} is largest for equally-sized clusters. 

Second, to maximize the domain of existence delimited by Eq.\eqref{eq:CC2}, we want the clusters to take as little space as possible on the circle. 
From Eq.~\eqref{eq.phi_cluster}, and using the convexity of arcsine, one has
\begin{align}\label{eq:bound}
N_c^{-1} \sum_{\alpha=1}^{N_c} ||\theta_{\alpha,+}^*-\theta_{\alpha,-}^*|| & \ge 
\arcsin\left(\omega_0 N_{c}^{-1} \sum_{\alpha=1}^{N_{c}} \frac{1}{f_\alpha} \right) \, .
\end{align}
We further use the inequality between harmonic and arithmetic means,~\cite{Bullen03}
\begin{align}
    N_{c}\left(\sum_{\alpha=1}^{N_{c}}\frac{1}{f_\alpha}\right)^{-1} \le  N_{c}^{-1} \left(\sum_{\alpha=1}^{N_{c}}f_\alpha \right) \, ,
\end{align}
to rewrite Eq.~\eqref{eq:bound} as 
\begin{align}\label{eq:bound2}
N_c^{-1} \sum_{\alpha=1}^{N_c} ||\theta_{\alpha,+}^*-\theta_{\alpha,-}^*|| & \ge \arcsin(2 \omega_0 N_{c})\, .
\end{align}
The final step is to realize that the right-hand-side of Eq.~\eqref{eq:bound2} gives the average cluster angle in the partition with equal-size clusters. 
Eq.~\eqref{eq:bound} states that it minimizes the average cluster angle, which concludes the proof.
\end{proof}

The phase diagram depends on partitions of the number of oscillators into clusters with even number of oscillators and
therefore depends on $N$. 
Because Eqs.~\eqref{eq:CC1} and \eqref{eq:CC2} depend on the half-fraction $f_\alpha$ of oscillators  
in each cluster, and not on $N$, we can nevertheless extract general, $N$-independent properties of the phase diagram
and of fixed-point stability.

First, a characteristic feature of all but one domain of existence is the nonmonotonicity of their upper demarcation line, which increases first to reach a maximal value of $\omega_0$ to then decrease, except for the single-cluster fixed point. 
This is easily understood for homogeneous fixed-points with equal-size clusters, which we explain here. 
In this case, the demarcation line rises for $0\leq\Delta\leq\Delta_{\max}$, following the condition given in Eq.~\eqref{eq:CC1} and reaches its maximal value for $\omega_0$ at $\Delta = \Delta_{\rm max}: = \frac{2\pi}{N_c}-\frac{\pi}{2}$. 
Then, for $\Delta_{\rm max}<\Delta\leq \frac{2\pi}{N_c}$, the domain is constrained by Eq.~\eqref{eq:CC2} and the demarcation curve goes down again in order to satisfy the equivalent relation 
\begin{align}\label{eq:explicitCC2}
 \omega_0 &< \frac{1}{2 N_c} \, \sin\left(\frac{2\pi}{N_c}-\Delta\right)\, .    
\end{align}
For the single-cluster fixed-point with $N_c=1$, however, $\Delta_{\rm max}>\pi$ and therefore, the demarcation curve remains
constant at $\omega_0=1/2$. 
This behavior is  illustrated in Fig. \ref{fig:CC1CC2} for $N_c=1,2$ and 3.

\begin{figure}
 \centering
 \includegraphics[width=.8\columnwidth]{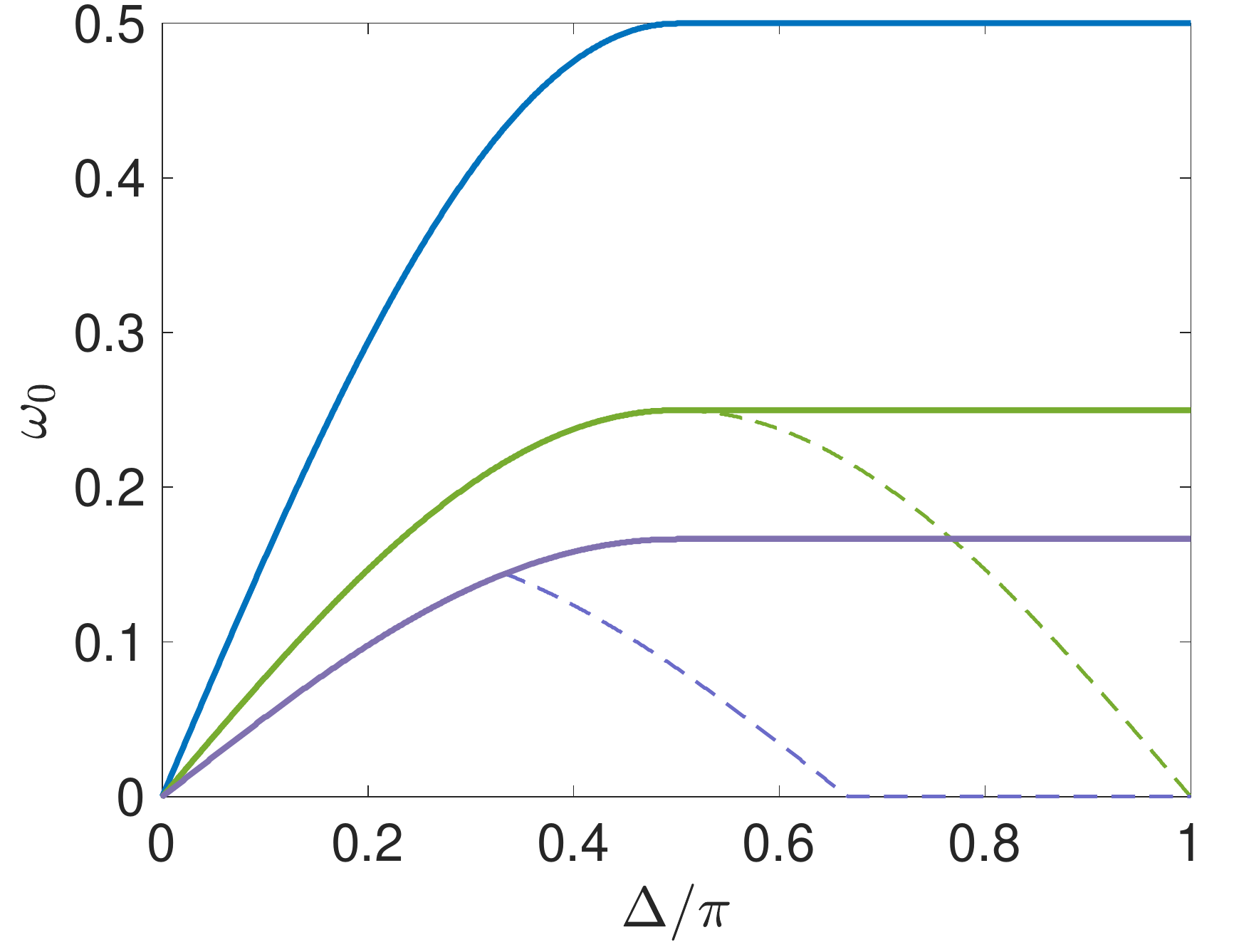}
 \caption{Illustration of the two conditions \eqref{eq:CC1} and \eqref{eq:CC2} for partition with $N_c=1$  (blue), $N_c=2$ (green), and $N_c=3$ (purple) cluster(s) of equal size. 
 Condition \eqref{eq:CC1} is satisfied below the plain line, and \eqref{eq:CC2} trivially satisfied for $\Delta \leq \Delta_{\max} = \frac{2\pi}{N_c}-\frac{\pi}{2}$ and is satisfied below the dashed line for $\Delta>\Delta_{\max}$. }
 \label{fig:CC1CC2}
\end{figure}

Second, the phase diagram will remain the same for $N= 12 a$, $a \in \mathbb{N}^+$ on larger scales, 
with additional domains emerging in the low-$\omega_0$, low-$\Delta$ regions of the phase diagram.
Furthermore, from Prop.~\ref{prop:largest domain of existence}, the partition that is the most stable at fixed number of clusters is the homogeneous one, where all clusters have the same size, if it exists. 

\subsection{Other distributions of natural frequencies}\label{ssec:other}
So far, we have considered a sharply peaked, symmetric distribution of natural frequencies $\omega_i = \pm \omega_0$ and one may wonder 
how much of our findings remain valid for more general distributions. Keeping the condition $\sum_i \omega_i = 0$ without loss of generality, 
the condition to form independent clusters is that $\sum_{i \in \Lambda} \omega_i = 0$ within each cluster $\Lambda$, see Eq.~\eqref{eq:sumcluster}.
Broadening the distribution of $\omega_i$'s, this conditions becomes harder to satisfy, which strongly limits the number of viable partitions. In particular 
if natural frequencies are randomly distributed in certain intervals, none of their partial sums will vanish and only the single cluster fixed-point 
survives. Nontrivial partitions leading to the rich phase diagram shown in Fig.~\ref{fig:EA} require nonrandom distributions of natural frequencies with 
vanishing partial sums.

\subsection{Opinion dynamics with changing rate}
Our analysis extends directly to the continuous-time opinion dynamics with changing rate
defined by
\begin{align}\label{eq:lin_bc}
 \dot{x}_i &= \omega_i - \frac{1}{N}\sum_{j=1}^Na_{ij}(x_i-x_j)\, , \\
 a_{ij} &= \left\{
 \begin{array}{ll}
  1\, , & \text{if } |x_i-x_j|<\Delta\, , \\
  0\, , & \text{otherwise.}
 \end{array}
 \right.
\end{align}
We consider two different cases. 

First, let us consider bounded opinion variables, $x_i \in [x_{\rm min},x_{\rm max}]$. 
Without loss of generality one may consider $x_{\rm min}=-\pi$ and $x_{\rm max}=\pi$. 
In that case, the discussion above remains the same, except for Eq.~\eqref{eq:cc_homog} which becomes
\begin{align}\label{eq:cc_homog_new}
 2\pi-\sum_{\alpha=1}^{N_c} |x_{\alpha,+}^*-x_{\alpha,-}^*| &> (N_c-1)\cdot\Delta\, ,
\end{align}
because boundary conditions are not periodic anymore.
The two conditions for the existence of a partition characterized by $\{ f_\alpha\}$ are given by
\begin{align}\label{eq:CC1_lin}
 \omega_0 &\leq \min_\alpha[f_\alpha]\cdot\Delta\, , & \Delta &\geq 0\, ,
\end{align}
instead of Eq.~\eqref{eq:CC1}, and 
\begin{align}\label{eq:CC2_lin}
  2\pi - \omega_0\sum_{\alpha=1}^{N_c}f_\alpha^{-1} &> (N_c-1)\cdot\Delta\, ,
\end{align}
instead of Eq.~\eqref{eq:CC2}. 
This leads to a similar phase diagram as in Fig.~\ref{fig:EA}, with demarcations being straight lines instead of curves. 

Second, one may consider unbounded opinion variables, $x_i \in \mathbb R$. 
In that case, clusters can be moved away from each other unboundedly and there is no restriction leading to the condition of Eq.~\eqref{eq:CC2_lin}. 
The phase diagram is then composed only of straight lines with the slope given by $\min_\alpha f_\alpha$.

\section{Conclusion and outlook}\label{conclusions}
Above we have introduced and given an extensive analysis of the dominant synchronous states of the bimodal Kuramoto model with bounded confidence.
The combination of nonlinear dynamics and state-dependent couplings renders the analysis of this model very challenging. 
Nevertheless, we manage to give a complete analysis of its synchronous fixed points when the interactions are all-to-all on noninteracting clusters of oscillators.
We have also showed that, even though other fixed point exist, they are less stable and have a smaller domain of existence in the parameter space. 
Consequently, our analysis covers the most relevant synchronous states of the model. 

In contrast to most works on clustering in bounded confidence models, our findings do not rely solely on numerical simulations. 
Our theory provides an analytic description of the fixed points of the Kuramoto model with bounded confidence, which is valid independently of the system size and parameter values. 
In particular, in the limit $\omega_0\searrow 0$, our results translate directly to the standard models of coupled dynamical agents with bounded confidence.~\cite{Lor07} 
It sheds an instructive light on the conditions under which clusters can emerge in models of opinion dynamics, emphasizing in particular the role of the confidence bound ($\Delta$) and of the self driving term ($\omega_0$). 
Furthermore, our stability analysis of different clustering structures gives an insight into the robustness of different opinion patterns against external disturbances. 
Assuming that models of opinion dynamics accurately represent some aspects of opinion formation, our results unravel the role of some social paramters in the construction and polarization of opinion in the population.

As stated in Sec.~\ref{ssec:other}, the richness of behavior observed in this manuscript derives mostly from our specific choice of natural frequencies. 
Further investigations should investigate other interesting distribution of natural frequencies. 

\section*{Acknowledgments}
This work has been supported by the Swiss National Science Foundation under grant 200020\_182050. 
RD acknowledges support from ETH Z\"urich funding. 


%

\end{document}